\documentclass[11pt]{article}
\usepackage{fullpage}
\usepackage{url}
\usepackage[colorlinks=true,linkcolor=red,citecolor=blue]{hyperref}
\usepackage{amsmath,amsfonts,amsthm,amssymb,multirow}
\usepackage{times}
\usepackage{paralist}
\usepackage{graphicx}
\usepackage{floatpag}
\usepackage{algorithm}
\usepackage[noend]{algpseudocode}
\usepackage{bbold}
\usepackage{enumitem}
\usepackage{subcaption} 
\usepackage{calc}
\usepackage{upgreek}
\usepackage{mathtools}
\usepackage{color}
\allowdisplaybreaks
\newenvironment{reminder}[1]{\bigskip
	\noindent {\bf Reminder of #1  }\em}{\smallskip}

\newtheorem{theorem}{Theorem}[section]

\newtheorem{corollary}{Corollary}[section]

\newtheorem{definition}{Definition}[section]
\newtheorem{lemma}{Lemma}[section]

\makeatletter
\let\c@fconjecture\c@conjecture
\makeatother

\makeatletter
\let\c@fconj\c@conj
\makeatother

\title{Super Strong ETH is False for Random $k$-SAT}
\author{Nikhil Vyas\footnote{\texttt{nikhilv@mit.edu}}\\MIT}
\date{}

\begin{document}
\maketitle
\begin{abstract}
It has been hypothesized that $k$-SAT is hard to solve for randomly chosen instances near the ``critical threshold'', where the clause-to-variable ratio is $2^k \ln 2-\theta(1)$. Feige's hypothesis~\cite{Feige02} for $k$-SAT says that for all sufficiently large clause-to-variable ratios, random $k$-SAT cannot be refuted in polynomial time. It has also been hypothesized that the worst-case $k$-SAT problem cannot be solved in $2^{n(1-\omega_k(1)/k)}$ time, as multiple known algorithmic paradigms (backtracking, local search and the polynomial method) only yield an $2^{n(1-1/O(k))}$ time algorithm. This hypothesis has been called the ``Super-Strong ETH'', modeled after the ETH and the Strong ETH.

Our main result is a randomized algorithm which refutes the Super-Strong ETH for the case of random $k$-SAT, for any clause-to-variable ratio. Given any random $k$-SAT instance $F$ with $n$ variables and $m$ clauses, our algorithm decides satisfiability for $F$ in $2^{n(1-\Omega(\log k)/k)}$ time, with high probability. It turns out that a well-known algorithm from the literature on SAT algorithms does the job: the PPZ algorithm of Paturi, Pudlak, and Zane (1998). 
\end{abstract}

\section{Introduction} The $k$-SAT problem is the canonical $NP$-complete problem for $k \geq 3$. No subexponential algorithms for $k$-SAT are known: in particular, no algorithms are known that achieve $2^{o(n)}$ time on $n$-variable $k$-SAT instances. Furthermore, despite much effort, no $(2-\epsilon)^n$-time algorithms for some constant $\epsilon > 0$ are known for $k$-SAT. The inability to find algorithms led researchers to the following two popular hypotheses:
\begin{compactitem}
\item \textbf{Exponential Time Hypothesis (ETH)}~\cite{eth} There is an $\alpha > 0$ such that no $3$-SAT algorithm runs in $2^{\alpha n}$ time.
\item 
\textbf{Strong Exponential Time hypothesis (SETH)}~\cite{seth} There does not exist a constant $\epsilon > 0$ such that for all $k$, $k$-SAT can be solved in  $(2-\epsilon)^n$ time.
\end{compactitem}

The current best known algorithms for $k$-SAT all have runtime $2^{n\left(1-O\left(\frac{1}{k}\right)\right)}$. This bound is achieved by multiple paradigms, such as randomized backtracking~\cite{ppz,ppsz}, local search~\cite{schoning}, and the polynomial method~\cite{ryan-ov}. Even for simpler variants such as unique-k-SAT, no faster algorithms are known. Hence it is possible that this runtime of $2^{n\left(1-O\left(\frac{1}{k}\right)\right)}$ is actually optimal. This has been termed the Super-Strong ETH~\cite{super-strong-eth}. 
 
 \textbf{Super-SETH: Super Strong exponential time hypothesis.} Super-SETH states that there are no $2^{n\left(1-\omega\left(\frac{1}{k}\right)\right)}$ algorithms for $k$-SAT.

There are related questions of a) finding solutions of a random $k$-SAT instance where each clause is drawn uniformly and independently from the set of all possible clauses and of b) finding solutions of a planted $k$-SAT instance where first a random (and hidden) solution $\sigma$ is sampled and then each clause is drawn uniformly and independently from the set of all possible clauses satisfying $\sigma$.  Random $k$-SAT displays a threshold behaviour in which, for $\alpha_{sat} = 2^k\ln{2}-\theta(1)$ and for all constants $\epsilon > 0$, the $k$-SAT instances are satisfiable with high probability when $m < (\alpha_{sat}-\epsilon)n$  and unsatisfiable with high probability when $m > (\alpha_{sat}+\epsilon)n$. Note that as far as decidability is concerned for instances below the threshold we can simply output satisfiable and above the threshold output unsatisfiable and we will be correct with high probability. It has been conjectured~\cite{cook, hard2} that random instances at the threshold $m = \alpha_{sat}n$ are the hardest random instances and it is hard to decide whether they are satisfiable or not. We are motivated by the following strengthening of this conjecture: Are random instances at the threshold as hard as the worst case instances of $k$-SAT?

\subsection{Prior Work}
As mentioned earlier there have been many algorithms for worst case $k$-SAT but none of them run in time $2^{n(1-\omega(1/k))}$. There also has been a lot of work on polynomial time algorithms for random $k$-SAT which return a solution for small $m$ below the threshold. Note that even though we know that these instances are satisfiable whp (with high probability) that does not immediately give us a way to find a solution.   
Chao and Franco\cite{chao} first proved that the unit clause heuristic(which is the same as PPZ) finds solutions with high probability for random $k$-SAT when $m \leq c2^{k}n/k$ for some constant $c > 0$. The current best known polynomial time algorithm in this regime is by Coja-Oghlan~\cite{better} and can find a solution with high probability for random $k$-SAT when $m \leq c2^{k}n\log{k}/k$ for some constant $c > 0$. Interestingly, we also know of polynomial time algorithms for large $m$. Specifically, it is known that for a certain constant $C_0 = C(k)$ and $m > C_0n$ there are polynomial time algorithms which find a solution for planted $k$-SAT by Krivelevich and Vilenchik~\cite{expected} and random $k$-SAT (conditioned on satisfiability) by Coja-Oghlan, Krivelevich and Vilenchik~\cite{almost}. But both of these results require at least $m > 4^kn/k$~\cite{dan}. To our knowledge, no improvements over worst case $k$-SAT algorithms are known for random $k$-SAT around the threshold. 

\subsection{Our Results}
We present an algorithm which breaks Super-Strong ETH for random $k$-SAT. In particular, we give a $2^{n\left(1-\Omega\left(\frac{\log{k}}{k}\right)\right)}$-time algorithm which finds a solution whp for random-k-SAT (conditioned on satisfiability) for all values of $m$. 
In fact, our algorithm is an old one from the SAT algorithms literature: the PPZ algorithm of Paturi, Pudlak and Zane~\cite{ppz}.

In order to show that PPZ breaks Super-Strong ETH in the random case, we first show that PPZ yields a faster algorithm for random \emph{planted} $k$-SAT for large enough $m$.

\begin{theorem} \label{thm:planted-algo}
There is a randomized algorithm that, given a planted $k$-SAT instance $F$ sampled from $P(n, k, m)$\footnote{See ``Three $k$-SAT Distributions" in Section~\ref{sec:pre} for formal definitions of different $k$-SAT distributions.} with $m > 2^{k-1}\ln(2)$, outputs a satisfying assignment to $F$ in $2^{n\left(1-\Omega\left(\frac{\log{k}}{k}\right)\right)}$ time with $1-2^{-\Omega \left(n\left(\frac{\log{k}}{k}\right)\right)}$ probability (over the planted $k$-SAT distribution and the randomness of the algorithm).
\end{theorem}

Next, we give a reduction from random $k$-SAT (conditioned on satisfiability, denoted by $R^+$) to planted $k$-SAT. Similar reductions/equivalences have been observed before in~\cite{eli, achlioptas-algorithmic}.

\begin{theorem}\label{thm:reduction} Suppose there exists an algorithm $A$ for planted k-SAT $P(n, k, m)$, for some $m\geq2^k\ln{2}(1-f(k)/2)n$, which finds a solution in time $2^{n(1-f(k))}$ and with probability $1-2^{-nf(k)}$, where $1/k < f(k) = o_{k}(1)$. Then given a random $k$-SAT instance sampled from $R^{+}(n, k, m)$, we can find a satisfiable solution in $2^{n(1-\Omega(f(k)))}$ time with $1-2^{-n\Omega(f(k))}$ probability.
\end{theorem}

Together, Theorems~\ref{thm:planted-algo} and \ref{thm:reduction}  imply the existence of a  $2^{n\left(1-\Omega\left(\frac{\log{k}}{k}\right)\right)}$ 
algorithm(See Theorem~\ref{thm:random-algo}) for finding solutions of random $k$-SAT (conditioned on satisfiability) whp for all $m$ and in particular this allows us to decide satisfiability for instance of random $k$-SAT at the threshold whp in the same running time. This can mean that either the random instances of k-SAT at the threshold are not the hardest instances of k-SAT or Super-Strong ETH is also false for worst case k-SAT. For PPZ algorithm $2^{n(1-O(\frac{1}{k}))}$ lower bounds~\cite{ppszlb} have been proven so now we know that with respect to the PPZ algorithm random $k$-SAT instances behave differently from worst case $k$-SAT instances.
On the other hand for PPSZ, the current best known algorithm for $k$-SAT (for $k \geq 4$) we only know $2^{n\left(1-O\left(\frac{\log{k}}{k}\right)\right)}$ lower bounds~\cite{ppszlb}, which matches our upper bounds, hence it is possible that PPSZ itself has $2^{n\left(1-\Omega\left(\frac{\log{k}}{k}\right)\right)}$ running time for worst case $k$-SAT. 

Finally in Section~\ref{sec:plantedlarge}, we see that our techniques can be used to get faster than $2^{n\left(1-\Omega\left(\frac{\log{k}}{k}\right)\right)}$ algorithms for planted $k$-SAT and random $k$-SAT (conditioned on satisfiability) when $m$ is large.

\section{Preliminaries}\label{sec:pre}
\paragraph*{Notation}
In this paper, we generally assume $k$ is a large enough constant. Our time bounds have the form $2^{n(1-\Omega(\log k)/k)}$, and we are beating $2^{n(1-O(1/k))}$ time; such notation really only makes sense for $k$ that can grow unboundedly.
We often use the terms ``solution'', ``SAT assignment'', and ``satisfying assignment'' interchangeably. The notation $x \in_r \chi$ denotes that $x$ is randomly sampled from the distribution $\chi$. By $poly(n)$ we mean some function $f(n)$ which satisfies $f(n) = O(n^c)$ for a fixed constant $c$. Letting $n$ be the number of variables in a $k$-CNF, a random event about them holds  \emph{whp} (with high probability) if it holds with probability $1-f(n)$ where $f(n) \rightarrow 0$ as $n \rightarrow \infty$. By $\log$ and $\ln$ we denote the logarithm function base $2$ and $e$ respectively. $H(p) = -p\log(p)-(1-p)\log(1-p)$ denotes the binary entropy function. $\tilde{O}(f(n))$ denotes $O(f(n)\log(f(n)))$

\paragraph*{Three $k$-SAT Distributions} In this paper we consider the following three distributions for $k$-SAT:
\begin{compactitem}
\item $R(n, k, m)$ is the distribution over formulas $F$ of $m$ clauses, where each clause is drawn i.i.d. from the set of all $k$-width clauses. This is the standard $k$-SAT distribution.
\item $R^{+}(n, k, m)$ is the distribution over formulas $F$ of $m$ clauses where each clause is drawn i.i.d. from the set of all $k$-width clauses and we condition $F$ on being satisfiable i.e. $R(n, k, m)$ conditioned on satisfiability.
\item $P(n, k, m, \sigma)$ is the distribution over formulas $F$ of $m$ clauses where each clause is drawn i.i.d. from the set of all $k$-width clauses which satisfy $\sigma$. $P(n, k, m)$ is the distribution over formulas $F$ formed by sampling $\sigma \in \{0, 1\}^n$ uniformly and then sampling $F$ from $P(n, k, m, \sigma)$.
\end{compactitem}

Note that an algorithm solving the search problem (finding SAT assignments) for instances sampled from $R^+$ is stronger than deciding satisfiability for instances sampled from $R$: given an algorithm for the search problem on $R^+$, we can run it on a random instance from $R$ and return SAT if and only if the algorithm returns a valid satisfying assignment.

\subsection{Structural properties of planted and random \texorpdfstring{$k$}{k}-SAT}
We will now state a few structural results about planted and random $k$-SAT which will be useful in proving the runtime and correctness of our algorithms. We prove some lemmas that bound the expected number of solutions of a planted $k$-SAT instance and a random $k$-SAT instance (conditioned on satisfiability).

The satisfiability conjecture for $k$-SAT states that satisfiability of random $k$-SAT displays a threshold behaviour for all $k$. The following lemma which states that it is true for all $k$ bigger than a fixed constant was proven by Ding, Sly and Sun\cite{k-sat-conjecture}.

\begin{lemma}[\cite{k-sat-conjecture}]\label{lem:k-sat-conj}
There exists a constant $k_0$ such that for all $k > k_0$, for $\alpha_{sat} = 2^k\ln{2}-\theta(1)$ and for all constant $\epsilon > 0$, we have that:
\begin{align*}
    \text{For }m < (1-\epsilon)\alpha_{sat}n, \lim_{n \to \infty} \Pr_{F \in_r R(n, k, m)}[\text{F is satisfiable}] =& 1\\
    \text{For }m > (1+\epsilon)\alpha_{sat}n, \lim_{n \to \infty} \Pr_{F \in_r R(n, k, m)}[\text{F is satisfiable}] =& 0
\end{align*}
\end{lemma}

We will also need the fact that whp the ratio of the number of solutions and its expected value is not too small as long as $m$ is not too large. This was proven by Achlioptas~\cite{achlioptas-algorithmic}.

\begin{lemma}[Lemma 22 in\cite{achlioptas-algorithmic}] \label{lem:expectedR}
For $F\in_rR(n, k, m)$, let $\mathcal{S}$ be the set of solutions of $F$. Then $E[|\mathcal{S}|] = 2^n(1-\frac{1}{2^k})^m$. Furthermore, for $\alpha_d = 2^k\ln{2}-k$ and $m < \alpha_dn$ we have that $$\lim_{n \to \infty} \Pr[|\mathcal{S}| < E[|\mathcal{S}|]/2^{O(nk/2^k)}] = 0$$.
\end{lemma}

\begin{lemma}\label{lem:expectedRplus}
For $F \in_r R^+(n, k, m)$ let $Z$ denote the number of solutions of $F$. Then for any constant $\delta > 0$, when $m < (1-\epsilon)\alpha_{sat}$ for some constant $\epsilon > 0$, then $2^n(1-\frac{1}{2^k})^m \leq E[Z] \leq (1+\delta)2^n(1-\frac{1}{2^k})^m$. Furthermore, for $\alpha_d=2^k\ln{2}-k$, and for $m < \alpha_dn$ we have that
$$\lim_{n \to \infty} \Pr[Z < E[Z]/2^{O(nk/2^k)}] = 0$$
\end{lemma}
\begin{proof}
Let $F' \in_r R(n, k, m)$ and let $Z'$ denote the number of solutions of $F'$. Let $p_n$ denote that probability that $F'$ is unsatisfiable, then $E[Z'] = (1-p_n)E[Z]$. By Lemma~\ref{lem:k-sat-conj} $\lim_{n \to \infty} p_n \to 0$, hence $2^n(1-\frac{1}{2^k})^m \leq E[Z] \leq (1+\delta)2^n(1-\frac{1}{2^k})^m$.

$\Pr[Z < E[Z]/2^{O(nk/2^k)}] \leq \Pr[Z' < E[Z]/2^{O(nk/2^k)}]$ as $Z$ is just $Z'$ conditioned on being positive. = $\Pr[Z' < E[Z]/2^{O(nk/2^k)}] \leq \Pr[Z' < E[Z']/2^{O(nk/2^k)}]$ as $E[Z] \leq 2E[Z']$. $\Pr[Z' < E[Z']/2^{O(nk/2^k)}]$ tends to $0$ by Lemma~\ref{lem:expectedR}.
\end{proof}

We will use our planted k-SAT algorithm to solve random k-SAT instances conditioned on their satisfiability. This approach was introduced in an unpublished manuscript by Ben-Sasson, Bilu, and Gutfreund ~\cite{eli}. We will use the following lemma therein.

\begin{lemma}[Lemma 3.3 from \cite{eli}] \label{lem:biased} For a given $F$ in $R^+(n, k, m)$ with $Z$ solutions, it is sampled from $P(n, k, m)$ with probability $Zp$, where $p$ only depends on $n$, $k$, and $m$.
\end{lemma}

\begin{corollary}\label{cor:planted-concentration} For $F \in_r R^+(n, k, m)$ and $F' \in_r P(n, k, m)$  let $Z$ and $Z'$ denote their number of solutions respectively. Then for $\alpha_d=2^k\ln{2}-k$ and for $m < \alpha_dn$, $\lim_{n \to \infty} \Pr[Z' < E[Z]/2^{O(nk/2^k)}] = 0$.
\end{corollary}
\begin{proof}
We have $\lim_{n \to \infty} \Pr[Z < E[Z]/2^{O(nk/2^k)}] = 0$ by Lemma~\ref{lem:expectedRplus}. Lemma~\ref{lem:biased} shows that the planted $k$-SAT distribution $P(n,k,m)$ is biased toward satisfiable formulas with more solutions. The distribution $R^+(n, k, m)$ instead chooses all satisfiable formulas with equal probability. Hence $\lim_{n \to \infty} \Pr[Z' < E[Z]/2^{O(nk/2^k)}] = 0$.
\end{proof}

Previous lemmas regarding the number of solutions do not apply when $m > \alpha_{sat}n$. Next we prove a lemma which bounds the number of expected solutions when $m > \alpha_{sat}n$ and this may be of independent interest.

\begin{lemma}\label{lem:expectedRplus2}
The expected number of solutions of $F \in_r R^+(n, k, m)$ and $F' \in_r P(n, k, m)$ for $m \geq (\alpha_{sat}-1)n$ are $\leq 2^{O(n/2^k)}$.
\end{lemma}
\begin{proof}
Lemma~\ref{lem:biased} shows that the planted $k$-SAT distribution $P(n,k,m)$ is biased toward satisfiable formulas with more solutions. Hence expected number of solutions of $F' \in_r P(n, k, m)$ upper bounds the expected number of solutions of $F \in_r R^+(n, k, m)$, so it suffices to upper bound expected number of solutions of $F'$.
Let $Z$ denote the number of solutions of $F'$. Let $\sigma$ denote the planted solution in $F$ and let $x$ be some assignment which at hamming distance $i$ from $\sigma$. For a clause $C$ which satisfies $\sigma$ but does not satisfy $x$ all the satisfying literals must come from the $i$ bits where $\sigma$ and $x$ differ and all the unsatisfying literals must come from the remaining $n-i$ bits. Let $j$ denote the number of satisfying literals in $C$ then probability that a randomly sampled clause $C$ which satisfies $\sigma$ and does not satisfy $x$ is $= \sum_{j=1}^{k} \frac{\binom{k}{j}}{2^k-1}(\frac{i}{n})^j(1-\frac{i}{n})^{k-j} = \frac{1-(1-\frac{i}{n})^k}{2^k-1}$. We will now upper bound $E[Z]$.

\begin{align*}
    E[Z] &= \sum\limits_{y \in \{0, 1\}^n} \Pr[\text{$y$  satisfies $F'$}]\\
    &= \sum\limits_{i = 1}^{n} \binom{n}{i} \Pr[\text{Assignment $x$ that differs from $\sigma$ in $i$ bits  satisfies $F'$}]\\
    &= \sum\limits_{i = 1}^{n} \binom{n}{i} \Pr[\text{A random clause satisfying $\sigma$ satisfies $x$}]^m\\
    &= \sum\limits_{i = 1}^{n} \binom{n}{i} (1-\Pr[\text{A random clause satisfying $\sigma$ does not satisfy $x$}])^m\\
    &= \sum\limits_{i = 1}^{n} \binom{n}{i} \left(1-\frac{1-\left(1-i/n\right)^k}{2^k-1}\right)^m\\
    &\leq \sum\limits_{i = 1}^{n} \binom{n}{i} e^{-m\left(\frac{1-\left(1-i/n\right)^k}{2^k-1}\right)} \hspace{10pt} \text{[As $1-x \leq e^{-x}$]}\\
    &\leq \sum\limits_{i = 1}^{n} \binom{n}{i} e^{-(\alpha_{sat}-1)n\left(\frac{1-\left(1-i/n\right)^k}{2^k-1}\right)}\\
    &\leq 2^{O(n/2^k)}\sum\limits_{i = 1}^{n} \binom{n}{i} e^{-((2^k-1)\ln{2})n\left(\frac{1-\left(1-i/n\right)^k}{2^k-1}\right)}\hspace{10pt} \text{[As $m \geq (2^{k}\ln{2}-O(1))n$]}\\
    &\leq 2^{O(n/2^k)}\sum\limits_{i = 1}^{n} \binom{n}{i} 2^{-n\left(1-\left(1-i/n\right)^k\right)}\\
    &\leq 2^{O(n/2^k)}\sum\limits_{i = 1}^{n} 2^{n\left(H(i/n)-1+\left(1-i/n\right)^k\right)}\\
    &\leq 2^{O(n/2^k)}\max\limits_{0 \leq p \leq 1} 2^{n\left(H(p)-1+\left(1-p\right)^k\right)}\\
\end{align*}
Let $f(p) = H(p)-1+\left(1-p\right)^k$. Then $f'(p) = -\log\left(\frac{p}{1-p}\right)-k(1-p)^{k-1}$ and $f''(p) = \frac{-1}{p(1-p)}+k(k-1)(1-p)^{k-2}$. $f''(p) = 0 \iff p(1-p)^{k-1} = \frac{1}{k(k-1)}$. Note that $f''(p)$ has only 2 roots in $[0, 1]$, hence $f'(p)$ has at most $3$ roots in $[0, 1]$. It can be verified that for large enough $k$, $f'(p)$ indeed has 3 roots at $p = \theta(1/2^k), \theta(\log{k}/k), 1/2-\theta(k/2^k)$. At all these 3 values of $p$, $f(p) = O(1/2^k)$. Hence $E[Z] \leq 2^{O(n/2^k)}$.
\end{proof}

\section{Planted k-SAT and the PPZ Algorithm}\label{sec:planted}

In this section, we establish that the PPZ algorithm solves random planted $k$-SAT instances faster than $2^{n-n/O(k)}$ time:

\begin{reminder}{Theorem~\ref{thm:planted-algo}} There is a randomized algorithm that given a planted $k$-SAT instance $F$ sampled from $P(n, k, m)$ with $m > 2^{k-1}\ln(2)$, outputs a satisfying assignment to $F$ in $2^{n\left(1-\Omega\left(\frac{\log{k}}{k}\right)\right)}$ time with $1-2^{-\Omega \left(n\left(\frac{\log{k}}{k}\right)\right)}$ probability (over the planted $k$-SAT distribution and the randomness of the algorithm).
\end{reminder}

We will actually prove the following stronger claim: For any $\sigma$, if $F$ was sampled from $P(n, k, m, \sigma)$, then we will find a set of $2^{n\left(1-\Omega\left(\frac{\log{k}}{k}\right)\right)}$ assignments in $2^{n\left(1-\Omega\left(\frac{\log{k}}{k}\right)\right)}$ time and with probability $1-2^{-\Omega \left(n\left(\frac{\log{k}}{k}\right)\right)}$ one of them will be $\sigma$ (the probability is over the planted $k$-SAT distribution and the randomness of the algorithm). Note that the above theorem statement implies an algorithm that (always) finds a solution for $k$-SAT instance $F$ sampled from $P(n, k, m)$ and runs in \emph{expected} time  $2^{n\left(1-\Omega(\frac{\log{k}}{k})\right)}$.

In fact, the algorithm of Theorem~\ref{thm:planted-algo} is a slightly modified version of the PPZ algorithm~\cite{ppz}, a well-known worst case algorithm for $k$-SAT. PPZ runs in polynomial time, and outputs a SAT assignment (on any satisfiable $k$-CNF) with probability $p \geq 2^{-n+n/O(k)}$. It can be repeatedly run for $\tilde{O}(1/p)$ times to obtain a worst-case algorithm that is correct whp. We consider a simplified version of the algorithm which is sufficient for analyzing planted $k$-SAT:

\begin{algorithm}[H]
\caption{Algorithm for planted $k$-SAT}\label{simple-ppz}
\begin{algorithmic}[1]
\Procedure{Simple-PPZ}{$F$}
\While {$i \leq n$}
\If {there exists a unit clause}
\State set the variable in it to make it true
\ElsIf {$x_i$ is unassigned}
\State Set $x_i$ randomly.
\State $i \gets i+1$
\Else \State $i \gets i+1$
\EndIf
\EndWhile
\State Output the assignment if it satisfies $F$.
\EndProcedure
\end{algorithmic}
\end{algorithm}

Our Simple-PPZ algorithm (Algorithm~\ref{simple-ppz}) only differs from PPZ in that the PPZ algorithm also performs an initial random permutation of variables. For us, a random permutation is unnecessary: a random permutation of the variables in the planted $k$-SAT distribution yields the same distribution of instances. That is, the original PPZ algorithm would have the same behavior as Simple-PPZ.

We will start with a few useful definitions.

\begin{definition}[\cite{ppz}] A clause $C$ is \emph{critical} with respect to variable $x$ and a satisfying assignment $\sigma$ if $x$ is the only variable in $C$ whose corresponding literal is satisfied by $\sigma$.
\end{definition}

\begin{definition} A variable $x_i$ in $F$ is \emph{good} for an assignment $\sigma$ if there exists a clause $C$ in $F$ which is critical with respect to $x$ and $\sigma$, and $i$ is largest index among all variables in $C$. A variable which is not good is called \emph{bad}.
\end{definition}

Observe that for every good variable $x_i$, if all variables $x_j$ for $j < i$ are assigned correctly with respect to $\sigma$, then Simple-PPZ will set $x_i$ to the correct value, due to the unit clause rule. 
As such, for a formula $F$ with $z$ good variables with respect to $\sigma$, the probability that Simple-PPZ finds $\sigma$ is at least $2^{-(n-z)}$: if all $n-z$ bad variables are correctly assigned, the algorithm is forced to set all good variables correctly as well. Next, we prove a high-probability lower bound on the number of good variables in a random planted $k$-SAT instance.

\begin{lemma}\label{lem:good}
A planted $k$-SAT instance $F$, sampled from $P(n, k, m, \sigma)$ with $m > n2^{k-1}\ln{2}$ has at least $\Omega(n\log{k}/k)$ good variables with probability $1-2^{-\Omega\left(\frac{n\log{k}}{k}\right)}$ with respect to the assignment $\sigma$.
\end{lemma}

\begin{proof}
Let $F \in_r P(n,k,m,\sigma)$ and let 
$L$ be the last (when sorted by index) $n\ln{k}/(2k)$ variables. Let $L_g, L_b$ be the good and bad variables respectively, with respect to $\sigma$, among $L$. Let $E$ denote the event that $|L_g| \leq n\ln{k}/(500k)$. Our goal is to prove a strong upper bound on the probability that $E$ occurs. For any $x_i \in L$, we have that, $i \geq n(1-\ln{k}/(2k))$. Suppose clause $C$ is good with respect to $x_i \in L_b$, then we know that $C$ does not occur in $F$. Next, we will lower bound the probability of such a clause occurring with respect to a fixed variable $x_i \in L$. Recall that in planted $k$-SAT each clause is drawn uniformly at random from the set of all clauses which satisfy $\sigma$. So we get that,

\begin{align*}
&\Pr[\text{$C$ is good with respect to $x_i \in L$}]\\
&= \frac{\text{Number of clauses which will make $x_i \in L$ good}}{\text{Total number of clauses which satisfy $\sigma$}}\\
&= \frac{\binom{i-1}{k-1}}{\binom{n}{k}(2^k-1)}\\ 
&\geq \frac{1}{2}\left(\frac{i}{n}\right)^{k-1}\frac{k}{2^kn} \hspace{10pt} \text{[As $i \geq n(1-\ln{k}/(2k))$]}\\ 
&\geq \frac{1}{2}\left(\frac{i}{n}\right)^{k}\frac{k}{2^kn}\\ 
&\geq \frac{1}{2}\left(1-\frac{\ln{k}}{2k}\right)^k\frac{k}{2^kn} \hspace{10pt}  \text{[As $i \geq n(1-\ln{k}/(2k))$]}\\ 
&\geq \frac{1}{2}\left(e^{-\ln{k}/k}\right)^k\frac{k}{2^kn} \hspace{10pt} \text{[As $k$ is a big enough constant and $e^{-w} \leq 1-w/2$ for small enough $w > 0$]}\\
&\geq \frac{1}{2^{k+1}n}
\end{align*}

If $E$ is true, then $|L_b| > n\ln{k}/(4k)$. So the probability of sampling a clause which make some variable $x_i \in L_b$ good is $\geq \frac{\ln{k}}{k2^{k+3}}$ as the set of clauses which make different variables good are disjoint. We will now upper bound the probability of $E$ occurring.

\begin{align*}
\Pr[E] &\leq \sum_{i=1}^{n\ln{k}/(500k)} \Pr[\text{Exactly $i$ variables among the last $n\ln{k}/(2k)$ variables are good}]\\
&\leq \sum_{i=1}^{ n\ln{k}/(500k)} \binom{n\ln{k}/(2k)}{i} \left(1-\frac{\ln{k}}{k2^{k+3}}\right)^{m}\\
&\leq n\binom{n\ln{k}/(2k)}{n\ln{k}/(500k)} \left(1-\frac{\ln{k}}{k2^{k+3}}\right)^{n2^{k-1}\ln{2}} \hspace{20pt} \text{[As $m > n2^{k-1}\ln{2}$]} \\
&\leq n\binom{n\ln{k}/(2k)}{n\ln{k}/(500k)} \left(e^{-\frac{\ln{k}}{k2^{k+3}}}\right)^{n2^{k-1}\ln{2}} \hspace{20pt} \text{[As $1-x \leq e^{-x}$ for $x > 0$]}\\
&\leq n\binom{n\ln{k}/(2k)}{n\ln{k}/(500k)} \left(2^{-\frac{n\ln{k}}{16k}}\right)\\
&\leq 2^{-\delta\frac{n\ln{k}}{k}} 
\end{align*}

for appropriately small but constant $\delta > 0$, which proves the lemma statement.

\end{proof}

We are now ready to prove Theorem~\ref{thm:planted-algo}.

\begin{proof}[Proof of Theorem~\ref{thm:planted-algo}] By Lemma~\ref{lem:good}, we know that with probability $\geq (1-p)$ for $p = 2^{-\Omega \left(n\left(\frac{\log{k}}{k}\right)\right)}$, the number of good variables with respect to hidden planted solution $\sigma$ in $F$ are $\geq \gamma n\log{k}/k$ for some constant $\gamma > 0$. For such instances one run of the PPZ algorithm will output $\sigma$ with probability $2^{-n(1-\gamma\log{k}/k)}$. Repeating the PPZ algorithm $\text{poly}(n)2^{n(1-\gamma\log{k}/k)}$ times, implies a success probability $\geq (1-p')$ for $p' = 2^{-n}$. Hence the overall error probability is at most $p+p' = 2^{-\Omega \left(n\left(\frac{\log{k}}{k}\right)\right)}$.
\end{proof}

We proved that PPZ algorithm runs in time $2^{n(1-\Omega(\frac{\log{k}}{k}))}$ when $m > n2^{k-1}\ln{2}$. For $m \leq 2^{k-1}\ln{2}$ we observe that the much simpler approach of randomly sampling works whp. This is because by Corollary~\ref{cor:planted-concentration} whp the number of solutions of $F \in_r P(n, k, m)$  for $m \leq 2^{k-1}\ln{2}$ is $\geq 2^{n/2}2^{-O(nk/2^k)}$. If so, then by randomly sampling $\text{poly(n)}2^{n/2}2^{O(nk/2^k)}$ assignments we will find a solution whp. As $m$ decreases further this sampling approach gives faster algorithms for finding a solution.

\section{Reductions from Random \texorpdfstring{$k$}{k}-SAT to Planted Random \texorpdfstring{$k$}{k}-SAT}\label{sec:random}
In this section we give a reduction from random k-SAT to planted k-SAT, which eventually gives an exponential algorithm for random k-SAT (see {Theorem~\ref{thm:random-algo}}). The following lemma is similar to the results in Achlioptas~\cite{achlioptas-algorithmic}, and we present it here for completeness.

\begin{lemma}[\cite{achlioptas-algorithmic}]\label{lem:small-reduction}
Suppose there exists an algorithm $A$ for planted k-SAT $P(n, k, m)$, for some $m\geq2^k\ln{2}(1-f(k)/2)n$, which finds a solution in time $2^{n(1-f(k))}$ and with probability $1-2^{-nf(k)}$, where $1/k < f(k) = o_{k}(1)$\footnote{We can assume wlog that $f(k) > 1/k$ as we already have a $2^{n(1-1/k)}$ algorithm for worst case $k$-SAT.}. Then given a random $k$-SAT instance sampled from $R^{+}(n, k, m)$, we can find a satisfiable solution in $2^{n(1-\Omega(f(k)))}$ time with $1-2^{-n\Omega(f(k))}$ probability.
\end{lemma}
\begin{proof}
Let $F$ be sampled from $R^+(n, k, m)$, and let $Z$ denote its number of solutions with $s$ its expected value. As $f(k) > 1/k$ and $m\geq2^k\ln{2}(1-f(k)/2)n$, Lemma~\ref{lem:expectedRplus} and \ref{lem:expectedRplus2} together imply that ${s \leq 2\cdot2^{nf(k)/2}}$.

We will now run Algorithm $A$. Note that if Algorithm $A$ gives a solution it is correct hence we can only have error when the formula is satisfiable but algorithm $A$ does not return a solution. We will now upper bound the probability of $A$ making an error.
\begin{align*}
&\Pr_{F \in R^+(n, k, m), A}[\text{A does not return a solution}]\\ \leq& \sum_{\sigma \in \{0, 1\}^n} \Pr_{F\in R^+(n, k, m), A}[\text{$\sigma$ satisfies F but $A$ does not return a solution}]\\
\leq& \sum_{\sigma \in \{0, 1\}^n} \Pr_{F\in R^+(n, k, m), A}[\text{$A$ does not return a solution} \mid \text{$\sigma$ satisfies F}]\Pr_{F\in R^+(n, k, m)}[\text{$\sigma$ satisfies F}]\\
\leq& \sum_{\sigma \in \{0, 1\}^n} \Pr_{F\in P(n, k, m, \sigma), A}[\text{$A$ does not return a solution}]\Pr_{F\in R^+(n, k, m)}[\text{$\sigma$ satisfies F}]
\end{align*}
where the last inequality used the fact that $R^+(n, k ,m)$ conditioned on having $\sigma$ as a solution is the distribution $P(n, k, m, \sigma)$.
Now note that $\Pr_{F\in R^+(n, k, m)}[\text{$\sigma$ satisfies F}] = s/2^n$ and $P(n, k, m)$ is just $P(n, k, m, \sigma)$ where $\sigma$ is sampled uniformly from $\{0, 1\}^n$. Hence the expression simplifies to 

$$= \frac{s}{2^n} (2^n\Pr_{F\in P(n, k, m), A}[\text{$A$ does not return a solution}]) = s\Pr_{F\in P(n, k, m), A}[\text{$A$ does not return a solution}]
$$
 As $s \leq 2\cdot2^{nf(k)/2}$ the error probability is $\leq 2\cdot2^{nf(k)/2}2^{-nf(k)} \leq 2\cdot2^{-nf(k)/2} = 2^{-\Omega(nf(k))}$.
\end{proof}

Next, we give another reduction from random $k$-SAT to planted $k$-SAT. This theorem is different from the previous one, in that, given a planted $k$-SAT algorithm that works in a certain regime of $m$, it implies a random $k$-SAT algorithm for \textit{all} values of $m$.

\begin{reminder}{Theorem~\ref{thm:reduction}} Suppose there exists an algorithm $A$ for planted k-SAT $P(n, k, m)$ for all $m\geq2^k\ln{2}(1-f(k)/2)n$, which finds a solution in time $2^{n(1-f(k))}$ and with probability $1-2^{-nf(k)}$, where $1/k < f(k) = o_{k}(1)$. Then for all $m'$, given a random $k$-SAT instance sampled from $R^{+}(n, k, m')$  we can whp find a satisfiable solution in $2^{n(1-\Omega(f(k)))}$ time.
\end{reminder}

\begin{proof}

Let $F$ be sampled from $R^+(n, k, m)$, and let $Z$ denote its number of solutions with $s$ its expected value. The expected number of solutions for $F'$ sampled from $R(n, k, m')$ serves as a lower bound for $s$. Hence if $m'\leq2^k\ln{2}(1-f(k)/2)n\leq\alpha_{d}n$, then $s > 2^{nf(k)/2}$ and furthermore, as we have $f(k) > 1/k$, Lemma ~\ref{lem:expectedRplus} implies that, $\lim_{n \to \infty} \Pr[ Z < s/2^{O(nk/2^k)}] = 0$. Hence, if we randomly sample $O(2^n2^{O(nk/2^k)}/s) = 2^{n(1-\Omega(f(k)))}$ assignments, one of them will satisfy $F$ whp. Otherwise if $m'\geq2^k\ln{2}(1-f(k)/2)n$ then we can use Lemma~\ref{lem:small-reduction} to solve it in required time.
\end{proof}

Now we combine Algorithm~\ref{simple-ppz} for planted $k$-SAT and the reduction in Theorem~\ref{thm:reduction}, to get an algorithm for finding solutions of random $k$-SAT (conditioned on satisfiability). This disproves Super-SETH for random $k$-SAT.
\begin{theorem}\label{thm:random-algo}
Given a random $k$-SAT instance $F$ sampled from $R^+(n, k, m)$ we can find a solution in $2^{n(1-\Omega(\frac{\log{k}}{k}))}$ time whp.
\end{theorem}
\begin{proof}
By Theorem~\ref{thm:planted-algo} we have an algorithm for planted $k$-SAT running in $2^{n(1-\Omega(\frac{\log{k}}{k}))}$ time with $1-2^{-\Omega \left(n\left(\frac{\log{k}}{k}\right)\right)}$ probability for all $m > (2^{k-1}\ln{2})n$. This algorithm satisfies the required conditions in Theorem~\ref{thm:reduction} with $f(k) = \Omega(\log{k}/k)$ for large enough $k$. The implication in Theorem~\ref{thm:reduction} proves the required statement.
\end{proof}

Just like in the case of planted $k$-SAT, we can whp find solutions of $R^+(n, k, m)$ by random sampling when $m < n(2^k\ln{2}-k)$ is smaller by just using random sampling. The correctness of random sampling follows from Lemma~\ref{lem:expectedRplus}.

\section{Planted and Random \texorpdfstring{$k$}{k}-SAT for large \texorpdfstring{$m$}{m} }\label{sec:plantedlarge}
In the previous sections we gave an algorithm that works at the threshold and also everywhere else. In this section we will work in the regime where $m$ is away from the threshold and give improved runtime analysis. As mentioned before current polynomial time algorithms that find solutions require $m$ to be at least $\frac{4^k}{k}n$. To our knowledge no algorithms are known for $2^kn < m < \frac{4^k}{k}n$ other than the worst case $k$-SAT algorithms. The proofs are similar to the proofs in Section~\ref{sec:planted} and~\ref{sec:random}.

\begin{lemma}\label{lem:good2}
A planted $k$-SAT $F$, instance sampled from $P(n, k, m, \sigma)$ with $2^{k+o(k)}n \geq m \geq 2^{k}n$ has at least $\Omega(nz)$ good variables with probability $1-2^{-\Omega(nz)}$ with respect to the assignment $\sigma$ where $z = (\ln(m/n)-k\ln{2})/k$.
\end{lemma}
\begin{proof}
In this proof by good/bad variables we will mean good/bad variables with respect to $\sigma$.

Let $F \in_r P(n,k,m,\sigma)$ and let 
$L$ be the last (when sorted by index) $nz/2$ variables. Let $L_g, L_b$ be the good and bad variables respectively, with respect to $\sigma$, among $L$. Let $E$ denote the event that $|L_g| \leq nz/500$. Our goal is to prove a strong upper bound on the probability that $E$ occurs. For any $x_i \in L$, we have that, $i \geq n(1-z/2)$. Suppose clause $C$ is good with respect to $x_i \in L_b$, then we know that $C$ does not occur in $F$. Next, we will lower bound the probability of such a clause occurring with respect to a fixed variable $x_i \in L$. Recall that in planted $k$-SAT each clause is drawn uniformly at random from the set of all clauses which satisfy $\sigma$. So we get that,

\begin{align*}
&\Pr[\text{$C$ is good wrt $x_i \in L$}]\\
&= \frac{\text{Number of clauses which will make $x_i \in L$ good}}{\text{Total number of clauses which satisfy $\sigma$}}\\ 
&= \frac{\binom{i-1}{k-1}}{\binom{n}{k}(2^k-1)}\\ 
&\geq \frac{1}{2}\left(\frac{i}{n}\right)^k\frac{k}{2^kn} \hspace{10pt} \text{[As $i \geq n(1-z/2), z = o(1)$]}\\ 
&\geq \frac{1}{2}\left(1-\frac{z}{2}\right)^k\frac{k}{2^kn} \hspace{10pt} \text{[As $i \geq n(1-z/2)$]}\\ 
&\geq \frac{1}{2}\left(e^{-z}\right)^k\frac{k}{2^kn}\hspace{10pt}  \text{[As $z = o(1)$ and $e^{-w} \leq 1-w/2$ for small enough $w > 0$]} \\
&\geq \frac{e^{-zk}}{2^{k+1}n}
\end{align*}

If $E$ is true, then $|L_b| > nz/4$. So the probability of sampling a clause which make some variable $x_i \in L_b$ good is $\geq \frac{ze^{-zk}}{2^{k+3}}$ as the set of clauses which make different variables good are disjoint. We will now upper bound the probability of $E$ occurring.

\begin{align*}
\Pr[E] &\leq \sum_{i=1}^{nz/500} \Pr[\text{Exactly $i$ good variables among the last $nz/2$ variables}]\\
&\leq \sum_{i=1}^{ nz/500} \binom{nz/2}{i} \left(1-\frac{ze^{-zk}}{2^{k+3}}\right)^{m}\\
&\leq n\binom{nz/2}{nz/500} \left(1-\frac{ze^{-zk}}{2^{k+3}}\right)^{ne^{zk}2^k} \hspace{10pt} \text{[As $m = e^{zk}2^kn$]}\\
&\leq n\binom{nz/2}{nz/500} \left(e^{-\frac{ze^{-zk}}{2^{k+3}}}\right)^{ne^{zk}2^k} \hspace{10pt} \text{[As $1-x \leq e^{-x}$ for $x > 0$]}\\
&\leq n\binom{nz/2}{nz/500} \left(e^{-\frac{nz}{8}}\right)\\
&\leq 2^{-\delta nz} 
\end{align*}

for appropriately small but constant $\delta > 0$, which proves the lemma statement.

\end{proof}

\begin{theorem}\label{thm:planted-algo2} Given a planted $k$-SAT instance $F$ sampled from $P(n, k, m)$ with $2^{k+o(k)}n > m > 2^kn$ define $z = (\ln(m/n)-k\ln{2})/k$ and $z' = z + \ln{k}/k$. Then we we can find a solution of $F$ in $2^{n(1-\Omega(z'))}$ time with $1-2^{-\Omega \left(nz'\right)}$ probability (over the planted $k$-SAT distribution and the randomness of the algorithm).
\end{theorem}
\begin{proof} By Lemma~\ref{lem:good2} we know that with probability $\geq (1-p)$ for $p = 2^{-\Omega \left(nz\right)}$ number of good variables in $F$ are $\gamma nz$ for some $\gamma > 0$ wrt the hidden planted solution $\sigma$. For such instances one run of the PPZ algorithm will output $\sigma$ with probability $2^{-n(1-\gamma z)}$. Repeating the PPZ $\text{poly(n})2^{n(1-\gamma z)}$ implies success probability of $\geq (1-p')$ for $p' = 2^{-n}$. The overall error probability is at most $p+p' = 2^{-\Omega \left(nz\right)}$.

Also by Theorem~\ref{thm:planted-algo} there exits an algorithm with $2^{n(1-\Omega(\frac{\log{k}}{k}))}$ time with $1-2^{-\Omega \left(n\left(\frac{\log{k}}{k}\right)\right)}$ probability. Both these algorithms together imply an algorithm running in $2^{n(1-\Omega(z'))}$ time with $1-2^{-\Omega \left(nz'\right)}$ probability (over the planted $k$-SAT distribution and the randomness of the algorithm).
\end{proof}

\begin{theorem}\label{thm:random-algo2} Given a random $k$-SAT instance $F$ sampled from $R^+(n, k, m)$ with $2^{k+o(k)}n > m > 2^kn$ define $z = (\ln(m/n)-k\ln{2})/k$ and $z' = z + \ln{k}/k$. Then we we can find a solution of $F$ in $2^{n(1-\Omega(z'))}$ time with $1-2^{-\Omega \left(nz'\right)}$ probability (over the random $k$-SAT distribution $R^+$ and the randomness of the algorithm).
\end{theorem}
\begin{proof}
This follows directly from composing the algorithm in Theorem~\ref{thm:planted-algo2} and the reduction in Lemma~\ref{lem:small-reduction}.
\end{proof}

As an example, the above theorem implies: For $F \in_r R^+(n, k, m)$ and $m = 2^{k+\sqrt{k}}n$ we have a $2^{n(1-\Omega(1/\sqrt{k}))}$ algorithm which works with $1-2^{-\Omega \left(n/\sqrt{k}\right)}$ probability.

Next we will increase $m$ even further and prove the existence of more good variables for the PPZ algorithm.

\begin{lemma}\label{lem:good3}
For planted $k$-SAT $F$, instance $F$ sampled from $P(n, k, m, \sigma)$ with $ m \geq t^{k}n$ where $t > 2$ is a constant. Then $F$ has at least $n(1-2/t)(1-2/k)$ good variables with probability $1-2^{-\Omega(n(1-2/t))}$ with respect to the assignment $\sigma$.
\end{lemma}

\begin{proof}
In this proof by good/bad variables we will mean good/bad variables wrt $\sigma$.

Let $F \in_r P(n,k,m,\sigma)$ and let 
$L$ be the last (when sorted by index) $nz$ variables where $z = 1-2/t$. Let $L_g, L_b$ be the good and bad variables respectively, with respect to $\sigma$, among $L$. Let $E$ denote the event that $|L_b| > \gamma nz$, where $\gamma = 2/k$. Our goal is to prove a strong upper bound on the probability that $E$ occurs. For any $x_i \in L$, we have that, $i \geq n(1-z)$. Suppose clause $C$ is good with respect to $x_i \in L_b$, then we know that $C$ does not occur in $F$. Next, we will lower bound the probability of such a clause occurring with respect to a fixed variable $x_i \in L$. Recall that in planted $k$-SAT each clause is drawn uniformly at random from the set of all clauses which satisfy $\sigma$. So we get that,

\begin{align*}
&\Pr[\text{$C$ is good with respect to $x_i \in L$}]\\
&= \frac{\text{Number of clauses which will make $x_i \in L$ good}}{\text{Total number of clauses which satisfy $\sigma$}}\\
&= \frac{\binom{i-1}{k-1}}{\binom{n}{k}(2^k-1)}\\ 
&\geq \frac{1}{2}\left(\frac{i}{n}\right)^k\frac{k}{2^kn} \hspace{10pt} \text{[As $i \geq n(1-z) = \Omega(n)$]}\\ 
&\geq \frac{1}{2}\left(1-z\right)^k\frac{k}{2^kn} \hspace{10pt} \text{[As $i \geq n(1-z)$]}\\ 
&= \frac{k\left(1-z\right)^k}{2^{k+1}n}
\end{align*}

If $E$ is true, then $|L_b| > \gamma nz$. So the probability of sampling a clause which make some variable $x_i \in L_b$ good is $\geq \frac{\gamma kz(1-z)^k}{2^{k+1}}$ as the set of clauses which make different variables good are disjoint. We will now upper bound the probability of $E$ occurring.

\begin{align*}
\Pr[E] &\leq \sum_{i=1}^{nz(1-\gamma)} \Pr[\text{Exactly $i$ good variables among the last $nz$ variables}]\\
&\leq \sum_{i=1}^{ nz(1-\gamma)} \binom{nz}{i} \left(1-\frac{\gamma kz\left(1-z\right)^k}{2^{k+1}}\right)^{m}\\
&\leq 2^{nz} \left(1-\frac{\gamma kz\left(1-z\right)^k}{2^{k+1}}\right)^{t^{k}n} \hspace{10pt} \text{[As $m > t^kn$]}\\
&\leq 2^{nz} \left(1-\frac{z\left(1-z\right)^k}{2^{k}}\right)^{t^{k}n} \hspace{10pt} \text{[$\gamma = 2/k$]}\\
&\leq 2^{n(1-2/t)} \left(1-\frac{(1-2/t)2^k}{t^k2^{k}}\right)^{t^{k}n} \hspace{10pt} \text{[Substituting value of $z$]}\\
&\leq 2^{n(1-2/t)} \left(1-\frac{(1-2/t)}{t^k}\right)^{t^{k}n}\\
&\leq 2^{n(1-2/t)} e^{-n(1-2/t)} \hspace{10pt} \text{[As $1-x \leq e^{-x}$ for $x > 0$]}\\
&\leq 2^{-\delta n(1-2/t)}
\end{align*}

for appropriately small but constant $\delta > 0$, which proves the lemma statement.

\end{proof}

\begin{theorem}\label{thm:planted-algo3} Given a planted $k$-SAT instance $F$ sampled from $P(n, k, m)$ with $ m \geq t^{k}n$ where $t > 2$ is a constant. Then we we can find a solution of $F$ in $2^{n(1-(1-2/t)(1-2/k))}\text{poly(n)}$ time with $1-2^{-\Omega(n(1-2/t))}$ probability (over the planted $k$-SAT distribution and the randomness of the algorithm).
\end{theorem}
\begin{proof} By Lemma~\ref{lem:good3} we know that with probability $\geq (1-p)$ for $p = 2^{-\Omega \left(n(1-2/t)\right)}$ number of good variables in $F$ is $\geq n(1-2/t)(1-2/k)$ with respect to the hidden planted solution $\sigma$. For such instances one run of the PPZ algorithm will output $\sigma$ with probability $2^{-n(1-(1-2/t)(1-2/k))}$. Repeating the PPZ $\text{poly(n)}2^{n(1-(1-2/t)(1-2/k))}$ implies constant success probability of $\geq 1-p'$ for $p' = 2^{-n}$. The overall error probability is at most $p+p' \leq 2^{-\Omega \left(n(1-2/t)\right)}$.
\end{proof}

For using Theorem~\ref{thm:planted-algo3} to get algorithms for $R^+$ we need a more refined version of Lemma~\ref{lem:small-reduction}.

\begin{lemma}\label{lem:small-reduction2}
Suppose there exists an algorithm $A$ for planted k-SAT $P(n, k, m)$ for some $m\geq \alpha_{sat}n$ which finds a solution in time $2^{n(1-f(k))}$ and with probability $p$. Then given a random $k$-SAT instance sampled from $R^{+}(n, k, m)$ we can find a satisfiable solution in $2^{n(1-f(k))}$ time with $1-(1-p)2^{O(n/2^k)}$ probability.

\end{lemma}
\begin{proof}
Let $F$ be sampled from $R^+(n, k, m)$, let $Z$ denote the number of solutions and $s$ its expected value. As $m\geq \alpha_{sat}n$ Lemma~\ref{lem:expectedRplus2} implies that $s \leq 2^{O(n/2^k)}$.

Now we will just run Algorithm $A$. Note that if Algorithm $A$ gives a solution it is correct hence we can only have error when the formula is satisfiable but algorithm $A$ does not return a solution.We will now upper bound the probability of $A$ making an error.
\begin{align*}
&\Pr_{F \in R^+(n, k, m), A}[\text{A does not return a solution}]\\ \leq& \sum_{\sigma \in \{0, 1\}^n} \Pr_{F\in R^+(n, k, m), A}[\text{$\sigma$ satisfies F but $A$ does not return a solution}]\\
\leq& \sum_{\sigma \in \{0, 1\}^n} \Pr_{F\in R^+(n, k, m), A}[\text{$A$ does not return a solution} \mid \text{$\sigma$ satisfies F}]\Pr_{F\in R^+(n, k, m)}[\text{$\sigma$ satisfies F}]\\
\leq& \sum_{\sigma \in \{0, 1\}^n} \Pr_{F\in P(n, k, m, \sigma), A}[\text{$A$ does not return a solution}]\Pr_{F\in R^+(n, k, m)}[\text{$\sigma$ satisfies F}]
\end{align*}
where the last inequality used the fact (refer) that $R^+(n, k ,m)$ conditioned on having $\sigma$ as a solution is exactly $P(n, k, m, \sigma)$.
Now note that $\Pr_{F\in R^+(n, k, m)}[\text{$\sigma$ satisfies F}] = s/2^n$ and $P(n, k, m)$ is just $P(n, k, m, \sigma)$ where $\sigma$ is sampled uniformly from $\{0, 1\}^n$. Hence the expression simplifies to 

$$= \frac{s}{2^n} (2^n\Pr_{F\in P(n, k, m), A}[\text{$A$ does not return a solution}]) = s\Pr_{F\in P(n, k, m), A}[\text{$A$ does not return a solution}]
$$
 As $s \leq 2^{O(n/2^k)}$ the error probability is $\leq 2^{O(n/2^k)}(1-p)$.
\end{proof}

\begin{theorem}\label{thm:random-algo3} Given a random $k$-SAT instance $F$ sampled from $R^+(n, k, m)$ with $ m \geq t^{k}n$ where $t > 2$ is a constant. Then we we can find a solution of $F$ in $2^{n(1-(1-2/t)(1-2/k))}\text{poly(n)}$ time with $1-2^{-\Omega(n(1-2/t))}$ probability (over the planted $k$-SAT distribution and the randomness of the algorithm).
\end{theorem}
\begin{proof}
The algorithm in Theorem~\ref{thm:planted-algo3} and the reduction in Lemma~\ref{lem:small-reduction2} imply that we can find a solution of $F$ in $2^{n(1-(1-2/t)(1-2/k))}\text{poly(n)}$ time with $1-2^{O(n/2^k)}2^{-\Omega(n(1-2/t))} = 1-2^{-\Omega(n(1-2/t))}$ probability. 
\end{proof}

\noindent\textbf{Acknowledgments.} The author would like to thank Ryan Williams for many invaluable discussions on the problem and help in writing the paper, Andrea Lincoln for discussions on Random $k$-SAT and Mitali Bafna for helpful comments on the paper.

\bibliographystyle{alpha}
\bibliography{references}
\end{document}